\documentclass[11pt]{article}
\usepackage{mathtools}

\usepackage[margin=1in]{geometry}
\usepackage{sn-preamble}
\usepackage{latexsym}
\usepackage{amsmath}
\usepackage{amssymb}
\usepackage{amsthm}
\usepackage{authblk}
\usepackage{qcircuit}
\usepackage{enumitem}
\usetikzlibrary{shapes.geometric}
\usetikzlibrary{patterns}
\usepackage{microtype}

\newcommand{\param}{2} 

\newcommand{\ignore}[1]{}

\newcommand{\rinn}{r_\mathrm{in}}
\newcommand{\Psymm}{{\cal P}_{\mathrm{symm}}}
\newcommand{\Pconv}{{\cal P}_{\mathrm{conv}}}
\newcommand{\Plcg}{{\cal P}_{\mathrm{slcg}}}
\newcommand{\Pltf}{{\cal P}_{\mathrm{LTF}}}
\newcommand{\Mix}{\mathrm{Mix}(\Psymm)}
\newcommand{\Mixgeneral}{\mathrm{Mix}(\Pconv)}
\newcommand{\Mixlcg}{\mathrm{Mix}(\Plcg)}
\newcommand{\dtv}{\mathrm{d}_{\mathrm{TV}}}

\newcommand{\Convex}{\textsc{Convex-Distinguisher}}
\newcommand{\symmconvex}{\textsc{Symm-Convex-Distinguisher}}

\def\Ball{\mathrm{Ball}}
\newcommand{\Poincare}{Poincar\'e}

\newcommand{\GMT}{\textsc{Gaussian-Mean-Testing}}
\newcommand{\dchi}{\mathrm{d}_{\chi^2}}
\renewcommand{\S}{\mathbb{S}}

\allowdisplaybreaks

\title{
Testing  Convex Truncation\thanks{A preliminary version of this paper appeared in the Proceedings of the 2023 Annual ACM-SIAM Symposium on Discrete Algorithms (SODA).}
\vspace{1em}}

\author{}

\date{}

\begin{document}

\pagenumbering{gobble}
\maketitle

\vspace{-6em}

{\large
\begin{center}

\begin{minipage}{0.32\textwidth}
	\centering
	Anindya De\\[0.25em]
	{University of Pennsylvania}
\end{minipage}
\begin{minipage}{0.32\textwidth}
	\centering
	Shivam Nadimpalli\\[0.25em]
	{MIT}
\end{minipage}
\begin{minipage}{0.32\textwidth}
	\centering
	Rocco A. Servedio\\[0.25em]
	{Columbia University}
\end{minipage}

\vspace{2em}

\today

\end{center}
}

\vspace{1em}

\begin{abstract}
We study the basic statistical problem of testing whether normally distributed $n$-dimensional data has been \emph{truncated}, i.e.~altered by only retaining points that lie in some unknown truncation set $S \subseteq \R^n$. 
As our main algorithmic results,
\begin{enumerate}

\item We give an $O(n)$-sample algorithm that can distinguish the standard normal distribution $N(0,I_n)$ from $N(0,I_n)$ conditioned on an unknown and arbitrary \emph{convex set} $S$.

\item We give a different $O(n)$-sample algorithm that can distinguish $N(0,I_n)$ from $N(0,I_n)$ conditioned on an unknown and arbitrary \emph{mixture of symmetric convex sets}.

\end{enumerate}
Both our algorithms are computationally efficient and run in $O(n^2)$ time, which is linear in the size of the input.

These results stand in sharp contrast with known results for learning or testing convex bodies with respect to the normal distribution or learning convex-truncated normal distributions, where state-of-the-art algorithms require essentially $n^{\sqrt{n}}$ samples. An easy argument shows that no finite number of samples suffices to distinguish $N(0,I_n)$ from an unknown and arbitrary mixture of general (not necessarily symmetric) convex sets, so no common generalization of results (1) and (2) above is possible.

We also prove that any algorithm (computationally efficient or otherwise) that can distinguish $N(0,I_n)$ from $N(0,I_n)$ conditioned on an unknown symmetric convex set must use $\Omega(n)$ samples. This shows that the sample complexity of each of our algorithms is optimal up to a constant factor.

\end{abstract}

\newpage
\pagenumbering{arabic}
\setcounter{page}{1}


\section{Introduction}
\label{sec:intro}

Understanding distributions which have been \emph{truncated}, i.e. subjected to some type of conditioning, is one of the oldest and most intensively studied questions in probability and statistics.  Research on truncated distributions goes back the work of Bernoulli \cite{Bernoulli60}, Galton \cite{Galton97}, Pearson \cite{Pearson02}, and other pioneers; we refer the reader to the introductions of \cite{DGTZ,Kontonis2019}  for historical context, and to \cite{Schneider86,BC14,Cohen16} for contemporary book-length studies of statistical truncation.

In recent years a nascent line of work \cite{DKTZcolt21,FKTcolt20,DGTZcolt19,DGTZ}
has considered various different learning and inference problems for truncated distributions from a modern theoretical computer science perspective (see \Cref{sec:related-work} for a more detailed discussion of these works and how they relate to the results of this paper).
The current paper studies an arguably more basic statistical problem than learning or inference, namely \emph{distinguishing} between a null hypothesis (that there has been no truncation) and an alternative hypothesis (that some unknown truncation has taken place).  

In more detail, we consider a high-dimensional version of the fundamental problem of determining whether given input data was drawn from a known underlying probability distribution  ${\cal P}$, versus from ${\cal P}$ conditioned on some unknown \emph{truncation set $S$} (we write ${\cal P}|_S$ to denote such a truncated distribution).  In our work the known high-dimensional distribution ${\cal P}$ is the $n$-dimensional standard normal distribution $N(0,I_n)$, and we consider a very broad and natural class of possible truncations, corresponding to conditioning on an unknown \emph{convex set} (and variations of this class).

As we discuss in detail in \Cref{sec:related-work}, the sample complexity and running time of known algorithms for a number of related problems, such as learning convex-truncated normal distributions \cite{Kontonis2019}, learning convex sets under the normal distribution \cite{KOS:08}, and testing whether an unknown set is convex under the normal distribution \cite{CFSS17}, all scale exponentially in $\sqrt{n}$.  
In sharp contrast, all of our distinguishing algorithms have sample complexity \emph{linear} in $n$ and running time at most $\poly(n)$.  Thus, our results can be seen as an exploration of one of the most fundamental questions in testing---namely, \emph{can we test faster than we can learn?} What makes our work different is that we allow the algorithm only to have access to random samples, which is weaker than the more powerful query access 
that is standardly studied in the complexity theoretic literature on property testing.  However, from the vantage point of statistics and machine learning, having only sample access is arguably more natural than allowing queries. Indeed, motivated by the work of Dicker~\cite{dicker2014variance} in statistics, a number of recent results in computer science~\cite{kong2018estimating, CDS20stoc, KongBV20} have explored the distinction between {\em testing versus learning} from random samples, and our work is another instantiation of this broad theme. 
To complement our algorithmic upper bounds, we also give a number of information theoretic lower bounds on sample complexity, which in some cases nearly match our algorithmic results.  We  turn to a detailed discussion of our results below.




\subsection{Our Results}

We give algorithms and lower bounds for a range of problems on distinguishing the normal distribution from various types of convex truncations.  

\subsubsection{Efficient  Algorithms}

Our most basic algorithmic result is an algorithm for symmetric convex sets:

\begin{theorem}[Symmetric convex truncations, informal statement]
 \label{thm:symmetric-convex-informal}
There is an algorithm \symmconvex~which uses $O(n/\eps^2)$ samples, runs in $O(n^2/\eps^2)$ time, and distinguishes between the standard $N(0,I_n)$ distribution and any distribution ${\cal D}=N(0,I_n)|_S$ where $S \subset \R^n$ is any symmetric convex set with Gaussian volume at most\footnote{Note that a Gaussian volume upper bound on $S$ is a necessary assumption, since the limiting case where the Gaussian volume of $S$ equals 1 is the same as having no truncation. In this case, it is information-theoretically impossible to detect truncation with any finite number of samples.}
 $1-\eps.$
\end{theorem}

The algorithm $\symmconvex$ is quite simple: it estimates the expected squared length of a random draw from the distribution and checks whether this value is significantly smaller than it should be for the $N(0,I_n)$ distribution. (See \Cref{sec:techniques} for a more thorough discussion of $\symmconvex$ and the techniques underlying its analysis.) By extending the analysis of $\symmconvex$  we are able to show that the same algorithm in fact succeeds for a broader class of truncations, namely truncation by any mixture of symmetric convex distributions:

\begin{theorem} 
[Mixtures of symmetric convex truncations, informal statement]
\label{thm:symmetric-convex-mixture-informal}
The algorithm \symmconvex~uses $O(n/\eps^2)$ samples, runs in $O(n^2/\eps^2)$ time, and distinguishes between the standard $N(0,I_n)$ distribution and any distribution ${\cal D}$ which is a normal distribution conditioned on a mixture of symmetric convex sets such that $\dtv(N(0,I_n),{\cal D}) \geq \eps$ (where $\dtv(\cdot,\cdot)$ denotes total variation distance).
\end{theorem}

It is not difficult to see that the algorithm $\symmconvex$, which only uses the empirical mean of the squared length of samples from the distribution, cannot succeed in distinguishing $N(0,I_n)$ from a truncation of $N(0,I_n)$ by a general (non-symmetric) convex set.  To handle truncation by general convex sets, we develop a different algorithm which uses both the estimator of $\symmconvex$ and also a second estimator corresponding to the squared length of the empirical mean of its input data points. We show that this algorithm succeeds for general convex sets:

\begin{theorem}[General convex truncations, informal statement]
\label{thm:general-convex-informal}
There is an algorithm \Convex\ which uses $O(n/\eps^2)$ samples, runs in $O(n^2/\eps^2)$ time, and distinguishes between the standard $N(0,I_n)$ distribution and any distribution ${\cal D}=N(0,I_n)|_S$ where $S \subset \R^n$ is any convex set such that $\dtv(N(0,I_n),N(0,I_n)|_S) \geq \eps.$
\end{theorem}

Given \Cref{thm:symmetric-convex-mixture-informal} and \Cref{thm:general-convex-informal}, it is natural to wonder about a common generalization to mixtures of general convex sets. However, an easy argument (which we sketch in \Cref{sec:no-generalization}) shows that no finite sample complexity is sufficient for this distinguishing problem, so no such common generalization is possible.

\ignore{
%
}

%

\subsubsection{An Information-Theoretic Lower Bound}

We show that the sample complexity of both our algorithms \Convex~and \symmconvex~are essentially the best possible, by giving an $\Omega(n/\eps)$-sample lower bound for any algorithm that successfully distinguishes $N(0,I_n)$ from $N(0,I_n)|_K$ where $K$ is an unknown symmetric convex set of volume $1-\eps$:

\begin{theorem}[Lower bound, informal statement]
\label{thm:halfspace-lb-informal}
For $\eps \in (0, 1/2]$, if $A$ is any algorithm which, given access to samples from an unknown distribution ${\cal D}$, successfully distinguishes the case that ${\cal D}=N(0,I_n)$ from the case that ${\cal D}=N(0,I_n)|_K$ where $K$ is an unknown symmetric convex set of volume $1-\eps$,  then $A$ must draw $\Omega\pbra{n/\eps}$ samples from $\calD$. 
\end{theorem}


\ignore{
%
%
%
}

\subsection{Techniques} \label{sec:techniques}

In this section, we give a technical overview of our upper and lower bounds, starting with the former.

\bigskip

\noindent {\bf Upper Bounds.}
To build intuition, let us first consider the case of a single symmetric convex body $K$. It can be shown, using symmetry and convexity of $K$, that draws from ${N}(0,I_n)|_K$ will on average lie closer to the origin than draws from ${N}(0,I_n)$, so it is natural to use this as the basis for a distinguisher. The proof of this relies on the background distribution being $N(0,I_n)$ in a crucial manner. We thus are led to consider our first estimator,
\begin{equation} \label{eq:e1}
\bM \coloneqq \frac{1}{T} {\sum_{i=1}^T \Vert\bx^{(i)} \Vert^2 },
\end{equation}
where $\bx^{(1)},\dots,\bx^{(T)}$ are independent draws from the unknown distribution (which is either $N(0,I_n)$ or $N(0,I_n)|_K$).
We analyze this estimator using the notion of \emph{convex influence} from the recent work \cite{DNS22}. In particular, we use a version of \Poincare's inequality for convex influence to relate the mean of $\bM$ to the Gaussian volume $\vol(K)$ of the truncation set $K$, and combine this with the fact that the statistical distance between $N(0,I_n)$ and $N(0,I_n)|_K$ is precisely $1-\vol(K)$.  
With some additional technical work in the analysis, this same tester turns out to work even for conditioning on a mixture of symmetric convex sets rather than a single symmetric convex set.


The estimator described above will not succeed for general (non-symmetric) convex sets; for example, if $K$ is a convex set that is ``far from the origin,'' then $\Ex_{\bx \sim N(0,I_n)|_K}[\|\bx\|]$ can be larger than $\Ex_{\bx \sim N(0,I_n)}[\|\bx\|]$. However, if $K$ is ``far from the origin,'' then the center of mass of a sample of draws from $N(0,I_n)|_K$ should be ``far from the origin,''  whereas the center of mass of a sample of draws from the standard normal distribution should be ``close to the origin;'' this suggests that a distinguisher based on estimating the center of mass should work for convex sets $K$ that are far from the origin.
%
The intuition behind our distinguisher for general convex sets is to trade off between the two cases that $K$ is ``far from the origin'' versus ``close to the origin.''  This is made precise via a case analysis based on whether or not the set $K$ contains a ``reasonably large'' origin-centered ball.\footnote{Splitting into these two cases is reminiscent of the case split in the analysis of a weak learning algorithm for convex sets in \cite{DS21colt}, though the technical details of the analysis are quite different in our work versus \cite{DS21colt}. In particular, \cite{DS21colt} relies on a ``density increment'' result for sets with large inradius, whereas we do not use a density increment argument but instead make crucial use of an extension of the Brascamp-Lieb inequality due to Vempala \cite{Vempala2010}.}

\medskip

\noindent {\bf Lower Bound.}  Our lower bound is proved by considering a randomly rotated ``symmetric slab'' of Gaussian volume $1-\eps$. Using this ``clean'' distribution for our lower bound construction makes it tractable to give a precise analysis of the  \emph{chi-squared divergence} between a sample of $T$ draws from the standard $n$-dimensional Gaussian conditioned on such a slab, versus a sample of $T$ draws from $N(0,I_n)$. 
The analysis reduces to getting good estimates on the noise stability of one-dimensional symmetric ``interval'' functions; the latter in turn can be expressed in terms of the Hermite spectrum of these one-dimensional functions. This lets us prove the stated lower bound.

\subsection{Related Work} \label{sec:related-work}

As noted earlier in the introduction, this paper can be viewed in the context of a recent body of work \cite{DKTZcolt21,FKTcolt20,DGTZcolt19,DGTZ,Kontonis2019} studying a range of statistical problems for truncated distributions from a theoretical computer science perspective.  In particular, \cite{DKTZcolt21} gives algorithms for non-parametric density estimation of sufficiently smooth multi-dimensional distributions in low dimension, while \cite{FKTcolt20} gives algorithms for parameter estimation of truncated product distributions over discrete domains, and \cite{DGTZcolt19} gives algorithms for truncated linear regression.  

The results in this line of research that are closest to our paper are those of \cite{DGTZ} and \cite{Kontonis2019}, both of which deal with truncated normal distributions (as does our work). \cite{DGTZ} considers the problem of inferring the parameters of an \emph{unknown} high-dimensional normal distribution given access to samples from a \emph{known} truncation set $S$, which is provided via access to an oracle for membership in $S$.  Note that in contrast, in our work the high-dimensional normal distribution is known to be $N(0,I_n)$ but the truncation set is unknown, and we are interested only in detecting whether or not truncation has occurred rather than performing any kind of estimation or learning.  Like \cite{DGTZ}, the subsequent work of \cite{Kontonis2019} considered the problem of estimating the parameters of an unknown high-dimensional normal distribution, but allowed for the truncation set $S$ to also be unknown.  They gave an estimation algorithm whose performance depends on the Gaussian surface area $\Gamma(S)$ of the truncation set $S$; when the set $S$ is an unknown convex set in $n$ dimensions, the sample complexity and running time of their algorithm is $n^{O(\sqrt{n})}$. In contrast, our algorithm for the distinguishing problem requires only $O(n)$ samples and $\poly(n)$ running time when $S$ is an unknown $n$-dimensional convex set. 

Other prior works which are related to ours are \cite{KOS:08} and \cite{CFSS17}, which dealt with Boolean function learning and property testing, respectively, of convex sets under the normal distribution.  \cite{KOS:08} gave an $n^{O(\sqrt{n})}$-time and sample algorithm for (agnostically) learning an unknown convex set in $\R^n$ given access to labeled examples drawn from the standard normal distribution, and proved an essentially matching lower bound on sample complexity.  \cite{CFSS17} studied algorithms for testing whether an unknown set $S \subset \R^n$ is convex versus far from every convex set with respect to the normal distribution, given access to random labeled samples drawn from the standard normal distribution.  \cite{CFSS17} gave an $n^{O(\sqrt{n})}$-sample algorithm and proved a near-matching $2^{\Omega(\sqrt{n})}$ lower bound on sample-based testing algorithms.  

We mention that our techniques are very different from those of \cite{DGTZ,Kontonis2019} and \cite{KOS:08,CFSS17}.
\cite{KOS:08} is based on analyzing the Gaussian surface area and noise sensitivity of convex sets using Hermite analysis, while \cite{CFSS17} uses a well-known connection between testing and learning \cite{GGR98} to leverage the \cite{KOS:08} learning algorithm result for its testing algorithm, and analyzes a construction due to Nazarov \cite{Nazarov:03} for its lower bound. \cite{DGTZ} uses a projected stochastic gradient descent algorithm on the negative log-likelihood function of the samples together with other tools from convex optimization, while (roughly speaking) \cite{Kontonis2019} combines elements from both \cite{KOS:08} and \cite{DGTZ} together with moment-based methods. In contrast, our approach mainly uses ingredients from the geometry of Gaussian space, such as the Brascamp-Lieb inequality and its extensions due to Vempala \cite{Vempala2010}, and the already-mentioned ``convex influence'' notion of \cite{DNS22}.

Finally, we note that the basic distinguishing problem we consider is similar in spirit to a number of questions that have been studied in the field of property testing of probability distributions \cite{Canonnedistributiontesting}. These are questions of the general form ``given access to samples drawn from a distribution that is promised to satisfy thus-and-such property, is it the uniform distribution or far in variation distance from uniform?''  Examples of works of this flavor include the work of Batu et al.~\cite{BKR:04} on testing whether an unknown monotone or unimodal univariate distribution is uniform; the work of Daskalakis et al.~\cite{DDSVV13} on testing whether an unknown $k$-modal distribution is uniform; the work of Rubinfeld and Servedio \cite{RubinfeldS09} on testing whether an unknown monotone high-dimensional distribution is uniform; and others.  The problems we consider are roughly analogous to these, but where the unknown distribution is now promised to be normal conditioned on (say) a convex set, and the testing problem is whether it is actually the normal distribution (analogous to being actually the uniform distribution, in the works mentioned above) versus far from normal.


\section{Preliminaries}
\label{sec:prelims}

In \Cref{subsec:useful-tools}, we set up basic notation and background. We recall preliminaries from convex and log-concave geometry in \Cref{subsec:convex-prelims,subsec:BL}, and formally describe the classes of distributions we consider in \Cref{subsec:dists-prelims}. 

\subsection{Basic Notation and Background}
\label{subsec:useful-tools}


\paragraph{Notation.}
We use boldfaced letters such as $\bx, \boldf,\bA$, etc. to denote random variables (which may be real-valued, vector-valued, function-valued, set-valued, etc.; the intended type will be clear from the context).
We write ``$\bx \sim \calD$'' to indicate that the random variable $\bx$ is distributed according to probability distribution $\calD.$ For $i\in[n]$, we will write $e_i\in\R^n$ to denote the $i^\text{th}$ standard basis vector.

\paragraph{Geometry.}
For $r >0$, we write $\S^{n-1}(r)$ to denote the origin-centered sphere of radius $r$ in $\R^n$ 
and $\Ball(r)$ to denote the origin-centered ball of radius $r$ in $\R^n$, i.e.,
$$
\S^{n-1}(r) = \big\{x \in \R^n: \|x\|=r\big\}\quad\text{and}\quad
\Ball(r) = \big\{x \in \R^n : \|x\| \leq r\big\},
$$
where $\|x\|$   denotes the $\ell_2$-norm $\|\cdot \|_2$ of $x\in \R^n$.
We also write $\S^{n-1}$ for the unit sphere $\S^{n-1}(1)$. 

Recall that a set $C \subseteq \R^n$ is convex if $x,y \in C$ implies $\alpha\hspace{0.03cm}x + (1-\alpha) y \in C$ for all $\alpha\in [0,1].$ 
Recall that convex sets are Lebesgue measurable.

For sets $A,B \subseteq \R^n$, we write $A + B$ to denote the Minkowski sum $\{a + b: a \in A\ \text{and}\ b \in B\}.$ For a set $A \subseteq \R^n$ and $r > 0$ we write $rA$ to denote the set $\{ra : a \in A\}$.
Given a point $a \in \R^n$ and a set $B\subseteq \R^n$, we use $a+B$ and $B-a$ to denote
  $\{a\}+B$ and $B+\{-a\}$ for convenience.  
\ignore{
}

\paragraph{Gaussians Distributions.}
We write $N(0,I_n)$ to denote the $n$-dimensional standard Gaussian distribution, and denote its density function by $\phi_n$, i.e. 
\[\varphi_n(x) = (2\pi)^{-n/2} e^{-\|x\|^2/2}.\]
When the dimension is clear from context, we may simply write $\phi$ instead of $\phi_n$. We write $\Phi:\R\to[0,1]$ to denote the cumulative density function of the one-dimensional standard Gaussian distribution, i.e. 
\[\Phi(x) := \int_{-\infty}^x \phi(y)\,dy.\]
We write $\vol(K)$ to denote the Gaussian volume of a (Lebesgue measurable) set $K \subseteq \R^n$, that is 
\[\vol(K) := \Prx_{\bx \sim N(0,I_n)}[\bx \in K].\]  
For a Lebesgue measurable  set $K \subseteq \R^n$, we write $N(0,I_n)|_K$ to denote the standard Normal distribution conditioned on $K$, so the density function of $N(0,I_n)|_K$ is 
\[\frac{1}{\vol(K)}\cdot\varphi_n(x)\cdot K(x)\]
where we identify $K$ with its $0/1$-valued indicator function. 
Note that the total variation distance between $N(0,I_n)$ and $N(0,I_n)|_K$ is
\begin{equation} \label{eq:tv-vol}
	\dtv\pbra{N(0,I_n)|_K,N(0,I_n)} = 1-\vol(K),
\end{equation}
and so the total variation distance between $N(0,I_n)$ and $N(0,I_n)|_K$ is at least $\eps$ if and only if $\vol(K) \leq 1-\eps.$

\paragraph{Gaussian Mean Testing.} We will require the following result due to Diakonikolas, Kane, and Pensia~\cite{DKP-SOSA}, which builds on prior work by Canonne et al.~\cite{SODACC}:

\begin{proposition}[Theorem~1.1 and Remark~1.2 of~\cite{DKP-SOSA}]
\label{prop:DKP}
	Let $\calD$ be a log-concave distribution over $\R^n$ and $\eps > 0$. There exists an algorithm, \GMT$(\calD, \eps)$, which, given i.i.d.~sample access to $\calD$, draws $\Theta(\max\{1, \sqrt{n}/\eps^2\})$ samples from $\calD$, does an $O(n^{3/2}/\eps^2)$-time computation, and has the following performance guarantee:
	\begin{itemize}
		\item If $\calD = N(0, I_n)$, then it outputs ``accept'' with probability $99/100$, and 
		\item If $\|\Ex[\bx]\| \geq \eps$ for $\bx\sim\calD$, then it will output ``reject'' with probability $99/100$. 
	\end{itemize}
\end{proposition}

\paragraph{Distinguishing Distributions.} We recall the basic fact that variation distance provides a lower bound on the sample complexity needed to distinguish two distributions from each other.

\begin{fact} [Variation distance distinguishing lower bound] \label{fact:dtv-distinguishing}
Let $P,Q$ be two distributions over $\R^n$ and let $A$ be any algorithm which is given access to independent samples that are either from $P$ or from $Q$.  If $A$ determines correctly (with probability at least $9/10$) whether its samples are from $P$ or from $Q$, then $A$ must use at least $\Omega(1/\dtv(P,Q))$ many samples.
\end{fact}

In order to prove our lower bound we will instead rely on the \emph{$\chi^2$-divergence}: Given two probability measures $P$ and $Q$ on $\R^n$ where $P$ is absolutely continuous with respect to $Q$ (i.e. for $S\sse\R^n$, $P(S) = 0$ whenever $Q(S) = 0$),
the $\chi^2$-divergence between $P$ and $Q$ is given by
\[
	\dchi(P\,\|\,Q) = \Ex_{\bx\sim Q}\sbra{\pbra{\frac{dP(\bx)}{dQ}}^2 - 1}.
\]
The following relationship between the $\chi^2$-divergence and the total variation distance is standard: 
\begin{equation} \label{eq:chi-to-tv}
	\dtv(P, Q)^2 \leq \dchi(P\,\|\,Q). 
\end{equation}

\subsection{Convex Influences}
\label{subsec:convex-prelims}

In what follows, we will identify a set $K\sse\R^n$ with its $0/1$-valued indicator function. The following notion of \emph{convex influence} was introduced in \cite{De2021,DNS22} as an analog of the well-studied notion of \emph{influence of a variable on a Boolean function} (cf. Chapter~2 of \cite{ODonnell2014}).  \cite{De2021,DNS22} defined this notion only for symmetric convex sets; we define it below more generally for arbitrary (Lebesgue measurable) subsets of $\R^n$.

\begin{definition}[Convex influence]
\label{def:influence}
	Given a Lebesgue measurable set $K \sse\R^n$ and a unit vector $v\in \S^{n-1}$, we define the \emph{convex influence of $v$ on $K$}, written $\Inf_v[K]$, as 
	\[\Inf_v[K] := \Ex_{\bx\sim N(0,I_n)}\sbra{K(\bx)\pbra{\frac{1 - \abra{v,\bx}^2}{\sqrt{2}}}}.\]
	Furthermore, we define the \emph{total convex influence of $K$}, written $\TInf[K]$, as 
	\[\TInf[K] := \sum_{i=1}^n \Inf_{e_i}[K] = \Ex_{\bx\sim N(0,I_n)}\sbra{K(\bx)\pbra{\frac{n - \|\bx\|^2}{\sqrt{2}}}}.\]
\end{definition}

In Proposition~20 of \cite{DNS22} it is shown that the influence of a direction $v$ captures the rate of change of the Gaussian measure of the set $K$ under a dilation along $v$. Also note that that total convex influence of a set is invariant under rotations. The following is immediate from \Cref{def:influence}.

\begin{fact}
\label{fact:avg-norm-influence}
	For Lebesgue measurable $K\sse\R^n$, we have
	\begin{equation} \label{eq:1-d-inf-formula}
		\Ex_{\bx\sim N(0,I_n)|_K}\sbra{\bx_i^2} = 1 - \frac{\sqrt{2}\cdot\Inf_{e_i}[K]}{\vol(K)}.
	\end{equation}
	We also have that
	\begin{equation} \label{eq:total-inf-estimator}
		\Ex_{\bx\sim N(0,I_n)|_K}\sbra{\|\bx\|^2} = n - \frac{\sqrt{2}\cdot\TInf[K]}{\vol(K)}.
	\end{equation}
\end{fact}

The following \Poincare-type inequality for convex influences was obtained as Proposition~23 in the full version of~\cite{DNS22} (see~\cite{DNS21influence22}). 

\begin{proposition}[\Poincare\ for convex influences for symmetric convex sets]
\label{prop:poincare}
	For symmetric convex $K\sse\R^n$, we have 
	\[\frac{\TInf[K]}{\vol(K)} \geq \Omega\pbra{1-\vol(K)}.\]
\end{proposition}

The following variant of \Cref{prop:poincare} for arbitrary convex sets (not necessarily symmetric) is implicit in the proof of Theorem~22 of \cite{DNS22} (see Equation~16 of \cite{DNS22}). Given a convex set $K\sse\R^n$, we denote its inradius by $\rinn(K)$, i.e. 
\[\rinn(K) := \max \cbra{r: \Ball(r)\sse K}.\]
When $K$ is clear from context, we will simply write $\rinn$ instead. 

\begin{proposition}[\Poincare\ for convex influences for general convex sets]
\label{prop:kkl}
	For convex $K\subseteq\R^n$ with $\rinn > 0$ (and hence $\vol(K)>0$), we have
	\[
		\frac{\TInf[K]}{\vol(K)} \geq \rinn\cdot\Omega\pbra{1-\vol(K)}.
	\]
\end{proposition}

\subsection{The Brascamp-Lieb Inequality}
\label{subsec:BL}

The following result of Brascamp~and~Lieb~\cite{BrascampLieb:76} generalizes the Gaussian \Poincare\ inequality to measures which are more log-concave than the Gaussian distribution.

\begin{proposition}[Brascamp-Lieb inequality]
\label{prop:BL}
	Let $\calD$ be a probability distribution on $\R^n$ with density $e^{-V(x)}\cdot\phi_n(x)$ for a convex function $V:\R^n\to\R$. Then for any differentiable function $f:\R^n\to\R$, we have 
	\[\Varx_{\bx\sim\calD}[f(\bx)]\leq \Ex_{\bx\sim\calD}\sbra{\|\nabla f(\bx)\|^2}.\]
\end{proposition}

Vempala~\cite{Vempala2010} obtained a quantitative version of \Cref{prop:BL} in one dimension, which we state next. Note in particular that the following holds for non-centered Gaussians.

\begin{proposition}[Lemma 4.7 of \cite{Vempala2010}]
\label{prop:vempala}
Fix $\theta\in\R$ and let $f: \R \to \R_{\geq  0}$ be a  log-concave function such that 
\[
\Ex_{\bx \sim N(\theta,1)}[\bx f(\bx)]=0.
\]
Then $\E[ \bx^2 f(\bx)] \le \E[f(\bx)]$ for $\bx\sim N(\theta,1)$, with equality if and only if $f$ is a constant function. Furthermore, if $\supp(f) \sse (-\infty,\eps]$, then 
\[\Ex_{\bx\sim N(\theta,1)}\sbra{\bx^2f(\bx)} \leq \pbra{1-\frac{1}{2\pi}e^{-\eps^2}}\Ex_{\bx\sim N(\theta,1)}\sbra{f(\bx)}.\]
\end{proposition}

\subsection{The Classes of Distributions We Consider}
\label{subsec:dists-prelims}

 We say that a distribution over $\R^n$ with density $\phi$ is \emph{symmetric} if $\phi(x)=\phi(-x)$ for all $x$, and that a set $K \subseteq \R^n$ is symmetric if $-x  \in K$ whenever $x \in K.$

We let $\Psymm$ denote the class of all distributions $N(0,I_n)|_K$ where $K \subseteq \R^n$ may be any symmetric convex set, $\Pconv$ denote the class of all such distributions where $K$ may be any convex set (not necessarily symmetric), and $\Pltf$ denote the class of all such distributions where $K$ may be any
linear threshold function $\sign(v \cdot x \geq {\theta}).$ We let $\Mix$ denote the class of all convex combinations (mixtures) of distributions from $\Psymm$, and we remark that a distribution in $\Mix$ can be viewed as $N(0,I_n)$ conditioned on a \emph{mixture} of symmetric convex sets.

The following alternate characterization of $\Mix$ may be of interest.  Let $\Plcg$ denote the class of all symmetric distributions that are log-concave relative to the standard normal distribution, i.e. all distributions that have a density of the form $e^{-\tau(x)} \varphi_n(x)$ where $\tau(\cdot)$ is a symmetric convex function. Let $\Mixlcg$ denote the class of all mixtures of distributions in $\Plcg.$

\begin{claim} 
$\Mixlcg=\Mix.$
\end{claim}
\begin{proof}
We will argue below that $\Plcg \subseteq \Mix.$ Given this, it follows that any mixture of distributions in $\Plcg$ is a mixture of distributions in $\Mix$, but since a mixture of distributions in $\Mix$ is itself a distribution in $\Mix$, this means that $\Mixlcg \subseteq \Mix.$ For the other direction, we observe that any distribution in $\Psymm$ belongs to $\Plcg$,\footnote{Recall that a distribution in $\Psymm$ has a density which is $\vol(K)^{-1}\cdot K(x)\cdot\varphi_n(x)$ for some symmetric convex $K$.} and hence $\Mix \subseteq \Mixlcg.$

Fix any distribution ${\cal D}$ in $\Plcg$ and let $e^{-\tau(x)}\varphi_n(x)$ be its density. We have that
\begin{equation} \label{eq:lolly}
e^{-\tau(x)} \varphi_n(x) = \E[A_{\bt}(x)] \cdot \varphi_n(x)
\end{equation}
where $A_t(x) = \Indicator[e^{-\tau(x)} \geq t]$ and the expectation in (\ref{eq:lolly}) is over a uniform $\bt \sim [0,1]$. Since $\tau$ is a symmetric convex function we have that the level set $\{x \in \R^n: e^{-\tau(x)} \geq t\}$ is a symmetric convex set, so ${\cal D}$ is a mixture of distributions in $\Psymm$ as claimed above.
\end{proof}




\section{An $O(n/\eps^2)$-Sample Algorithm for Symmetric Convex Sets and Mixtures of Symmetric Convex Sets}

In this section, we give an algorithm (cf. \Cref{alg:symconvex}) to distinguish Gaussians from (mixtures of) Gaussians truncated to a symmetric convex set. 

\subsection{Useful Structural Results}

We record a few important lemmas which are going to be useful for the analysis in this section. 

\begin{lemma}~\label{lem:bound-on-expectation}
Let $K \subseteq \mathbb{R}^n$ be a centrally symmetric convex set.  If $\vol(K) \le 1-\epsilon$, then, 
\[
\Ex_{\bx \sim N(0,I_n)|_K} [\Vert \bx \Vert^2] \le  n - c \epsilon
\]
for some absolute constant $c>0.$
\end{lemma}

\begin{proof}
We have
\[
		\Ex_{\bx\sim N(0,I_n)|_K}\sbra{\|\bx\|^2} = n - \frac{\sqrt{2}\cdot\TInf[K]}{\vol(K)}
		\leq
		n - \sqrt{2} \cdot c'(1 - \vol(K))
		\leq
		n - \sqrt{2} \cdot c'\eps,
\]
where the equality is \Cref{eq:total-inf-estimator}, the first inequality is \Cref{prop:poincare} (\Poincare\ for convex influences for symmetric convex sets), and the second inequality holds because $\vol(K) \le 1-\epsilon$. 
\end{proof}


\begin{lemma}~\label{lem:bound-variance-direction}
Let $K \subseteq \mathbb{R}^n$ be a convex set (not necessarily symmetric) and let ${\cal D} = N(0,I_n)|_K$.
Then for any unit vector $v$, we have 
\[\Varx_{\bx \sim {\cal D}} [v \cdot \bx] \le 1.\]
\end{lemma}
\begin{proof}
Given $c>0$, we define $V_c: \R^n \to \{c,+\infty\}$ to be
\[
V_c(x)=
\begin{cases}
c & \text{~if~}x \in K\\
+\infty &  \text{~if~}x \notin K.
\end{cases}
\]
We note that $V_c(\cdot)$ is a convex function for any choice of $c>0$, and that for a suitable choice of $c$, the density function of ${\cal D}$ is $e^{-V_c(x)} \cdot \gamma_n(x)$. Thus, we can apply the Brascamp-Lieb inequality to get that for any differentiable $f: \mathbb{R}^n \rightarrow \mathbb{R}$, 
\begin{equation}~\label{eq:Brascamp-convex-body}
\Varx_{\bx \sim \mathcal{D}} [f(\bx)] \le \Ex_{\bx \sim \mathcal{D}} [\Vert \nabla f(\bx) \Vert^2]. 
\end{equation}
Now, we may assume without loss of generality that $v=e_1.$
Taking $f(x) = x_1$ in \Cref{eq:Brascamp-convex-body}, we get that
\[
\Varx_{\bx \sim \mathcal{D}} [\bx_1] \le 1,
\]
which finishes the proof. 
\end{proof}

Now we can bound the variance of $\Vert \bx\Vert^2$ when $\bx \sim N(0,I_n)|_K$ for a symmetric convex set $K$. 

\begin{lemma}~\label{lem:bound-variance}
Let ${\cal D} = N(0,I_n)|_K$ for a symmetric convex set $K$. Then, $\Varx_{\bx \sim {\cal D}} [\Vert \bx \Vert^2] \le 4n$. 
\end{lemma}
\begin{proof}
Taking $f (x) := \Vert x \Vert^2$ in \Cref{eq:Brascamp-convex-body}, we have that 
\[
\Varx_{\bx \sim \mathcal{D}} [\Vert \bx \Vert^2] \le 4 \cdot 
\Ex_{\bx \sim \mathcal{D}} [{\bx_1^2 + \ldots + \bx_n^2}].  
\]
Since $K$ is symmetric, for each $i \in [n]$ we have $\Ex_{\bx \sim {\cal D}}[\bx_i]=0$ and hence $\Ex_{\bx \sim {\cal D}}[\bx_i^2] = \Var[e_i \cdot \bx]$, which is at most 1 by \Cref{lem:bound-variance-direction}.
\end{proof}

\subsection{An ${O}(n/\eps^2)$-Sample Algorithm for Symmetric Convex Sets}

We recall \Cref{thm:symmetric-convex-informal}:

\begin{theorem} [Restatement of \Cref{thm:symmetric-convex-informal}]\label{thm:symmetric}
For a sufficiently large constant $C>0$, the algorithm \symmconvex~(\Cref{alg:symconvex}) has the following performance guarantee:  given any $\eps > 0$ and access to independent samples from any unknown distribution ${\cal D} \in \Psymm$, the algorithm uses $Cn/\eps^2$ samples, and

\begin{enumerate}

\item If ${\cal D}=N(0,I_n)$, then with probability at least $9/10$ the algorithm outputs ``un-truncated'';

\item If $\dtv({\cal D},N(0,I_n)) \geq \eps$, then with probability at least $9/10$ the algorithm outputs ``truncated.''

\end{enumerate}
\end{theorem}

As alluded to in \Cref{sec:techniques}, \symmconvex\ uses the estimator from \Cref{eq:e1}.

\begin{algorithm}
\caption{Distinguisher for (Mixtures of) Symmetric Convex Sets}
\label{alg:symconvex}
\vspace{0.5em}
\textbf{Input:} $\calD\in\Pconv$, $\eps > 0$\\[0.5em]
\textbf{Output:} ``Un-truncated'' or  ``truncated''

\ 

\symmconvex$(\calD, \epsilon)$:
\begin{enumerate}
    \item For $T =C \cdot n/\epsilon^2$, sample points $\bx^{(1)}, \ldots, \bx^{(T)} \sim \calD$. 
    \item Let $\bM \coloneqq {\frac 1 T} \sum_{i=1}^T \Vert\bx^{(i)} \Vert^2 $. 
    \item If $\bM\ge n-c\epsilon/2$, output ``un-truncated," else output ``truncated". 
\end{enumerate}

\end{algorithm}


\begin{proofof}{\Cref{thm:symmetric}}
Let ${\cal D}_{G} := N(0,I_n)$ and 
${\cal D}_{T} := N(0,I_n)|_K$. 
 Then, for $\bx \sim {\cal D}_{G}$, the random variable $\Vert \bx \Vert^2$ follows the $\chi^2$ distribution with $n$ degrees of freedom, and thus we have
 \begin{equation}~\label{eq:stats-for-Gaussian}
\Ex_{\bx \sim {\cal D}_{G} } [\Vert \bx \Vert^2] =n; \ \ \Varx_{\bx \sim {\cal D}_{G} } [\Vert \bx \Vert^2] = 3n. 
 \end{equation}
On the other hand, if $\dtv({\cal D},N(0,I_n)) \geq \eps$ (equivalently, $\vol(K) \leq 1-\eps$), then using \Cref{lem:bound-on-expectation} and \Cref{lem:bound-variance}, it follows that 
\begin{equation}~\label{eq:stats-for-truncated}
 \Ex_{\bx \sim {\cal D}_{T} } [\Vert \bx \Vert^2] \le n-c\epsilon; \ \ \Varx_{\bx \sim {\cal D}_{T} } [\Vert \bx \Vert^2] \le 4n. 
 \end{equation}
Since in~\Cref{alg:symconvex} the samples $\bx^{(1)}, \ldots, \bx^{(T)}$ are independent, we have the following:
\begin{align*}
	\mathbf{E}[\bM ] = n \quad\text{and}\quad \Var[\bM] = \frac{3n}{T} \qquad \text{when } \calD = \calD_G,\\
	\mathbf{E}[\bM ] = n-c\epsilon \quad\text{and}\quad \Var[\bM] \leq \frac{4n}{T} \qquad \text{when } \calD = \calD_T.
\end{align*}
 By choosing $T = Cn/\epsilon^2$ (for a sufficiently large constant $C$), it follows that when ${\cal D} = {\cal D}_{{G}}$ (resp. ${\cal D} = {\cal D}_{{T}}$), with probability at least $9/10$ we have $\bM \ge n-c\epsilon/2$ (resp. $\bM < n-c\epsilon/2$). This finishes the proof. 
\end{proofof}

\subsection{An ${O}(n/\eps^2)$-Sample Algorithm for Mixtures of Symmetric Convex Sets}

By extending the above analysis, we can show that \Cref{alg:symconvex} succeeds for mixtures of (an arbitrary number of) symmetric convex sets as well. In particular, we have the following:

\begin{theorem} \label{thm:mix}
For  a sufficiently large constant $C>0$, \symmconvex~(\Cref{alg:symconvex}) has the following performance guarantee:  given any $\eps > 0$ and access to independent samples from any unknown distribution ${\cal D} \in \Mix$, the algorithm uses $Cn/\eps^{\param}$ samples, and

\begin{enumerate}

\item If ${\cal D}=N(0,I_n)$, then with probability at least $9/10$ the algorithm outputs ``un-truncated'';

\item If $\dtv({\cal D},N(0,I_n)) \geq \eps$, then with probability at least $9/10$ the algorithm outputs ``truncated.''

\end{enumerate}

\end{theorem}

The following lemma, which characterizes the mean and variance of a distribution in $\Mix$ in terms of the components of the mixture, will crucial to the proof of \Cref{thm:mix}:

\begin{lemma}~\label{lem:calc-for-mix}
Let $\mathcal{X}$ denote a distribution over Gaussians truncated by symmetric convex sets. Suppose ${\cal D_{\mathcal{X}}} \in \Mix$ is the mixture of $N(0,I_n)|_\bK$ for $\bK \sim \mathcal{X}$. 
Let $\ba_{\bK}$ denote the random variable 
\[\ba_{\bK} = \Ex_{\bx \sim  N(0,I_n)|_{\bK}} \sbra{\Vert \bx \Vert^2} \qquad \text{where }\bK \sim \mathcal{X}.\] 
Then 
\begin{equation} \label{eq:expectation-ak}
	\Ex_{\bx \sim {\cal D}_{\mathcal{X}}} \sbra{\Vert \bx \Vert^2} = \Ex_{\bK \sim \mathcal{X}} \sbra{\ba_{\bK}},
\end{equation}
\begin{equation} \label{eq:variance-ak}
	\Varx_{\bx \sim {\cal D}_{\mathcal{X}}} \sbra{\Vert \bx \Vert^2} \leq 4n + \Varx_{\bK \sim \mathcal{X}} \sbra{\ba_{\bK}}.
\end{equation}
\end{lemma}
\begin{proof}
Note that \Cref{eq:expectation-ak} follows from linearity of expectation and the definition of $\ba_{\bK}$. For \Cref{eq:variance-ak}, note that for any symmetric convex set $K$, by definition of variance we have 
\begin{align*}
	\Ex_{\bx \sim N(0,I_n)|_K} \sbra{\Vert \bx \Vert^4} &= \left(\Ex_{\bx \sim N(0, I_n)|_K} \sbra{\Vert \bx \Vert^2} \right)^2 + \Varx_{\bx \sim N(0, I_n)|_K} \sbra{\Vert \bx \Vert^2} \nonumber \\ 
&\le \boldsymbol{a}_K^2 + 4n,
\end{align*}
where the inequality is by~\Cref{lem:bound-variance}. By linearity of expectation, it now follows that
\[
\Ex_{\bx \sim {\cal D_{\mathcal{X}}}} \sbra{\Vert \bx \Vert^4} \le 4n + \Ex_{\bK \sim \mathcal{X}} \sbra{\ba_{\bK}^2}. 
\] 
Combining with \Cref{eq:expectation-ak}, we get \Cref{eq:variance-ak}. 
\end{proof}
We are now ready to prove ~\Cref{thm:mix}.
\begin{proofof}{\Cref{thm:mix}} 
Let $\mathcal{X}$ denote a distribution over symmetric convex sets. Define ${\cal D_{\mathcal{X}}} \in \Mix$ to be the mixture of $N(0,I_n)|_{\bK}$ for $\bK \sim \mathcal{X}$ and define ${\cal D}_{{G}} := N(0,I_n)$. Using the fact that the samples $\bx^{(1)}, \ldots, \bx^{(T)}$ are independent, as in the proof of~\Cref{thm:symmetric}, 
we have that
\begin{equation} \label{eq:Gaussian-statistics}
	\E[\bM] = n, \quad \Var[\bM] = \frac{3n}{T} \qquad \text{when }\calD = \calD_G.
\end{equation}
As $T = Cn/\epsilon^{\param}$ (for a sufficiently large constant $C$), it follows that when ${\cal D} = {\cal D}_{{G}}$, with probability at least $9/10$ we have that $\bM \ge n-\epsilon/2$. 
 
 Now we analyze the case that ${\cal D} = {\cal D_{\mathcal{X}}}$ has $\dtv({\cal D},N(0,I_n)) \geq \eps$.
 From~\Cref{lem:calc-for-mix}, it follows that in this case
 \begin{equation}
 	\mathbf{E} \sbra{\bM} = \Ex_{\bK \sim \mathcal{X}} \sbra{\ba_{\bK}},  \label{eq:truncated-mix-statistics1}
 \end{equation}
 \begin{equation}
 	\Varx[\bM] = \frac{
\Varx_{\bx \sim {\cal D}_{\mathcal{X}}} \sbra{\Vert \bx \Vert^2} }{T} \leq
 \frac{4n}{T} + \frac{\Varx_{\bK \sim \mathcal{X}} \sbra{\ba_{\bK}}}{T}.
   \label{eq:truncated-mix-statistics}
 \end{equation}
  Next, observe that 
 \begin{equation}~\label{equation:n-ak}
\Ex_{\bK \sim \mathcal{X}} \left[(n-\ba_{\bK})\right] \ge 
c \cdot \Ex_{\bK \sim \mathcal{X}} \left[1-\vol(\bK)\right] \ge 
c \cdot \dtv({\cal D}, N(0,I_n)) \ge c\epsilon,
 \end{equation}
 where the first inequality uses~\Cref{lem:bound-on-expectation} and the second inequality follows from the definition of TV distance.  
 Now, observing that variance of a random variable is invariant under negation and  translation and that  $T = Cn/\epsilon^\param$, it follows 
from \Cref{eq:truncated-mix-statistics} that 
\[
 \Var[\bM] \le \frac{4n}{T} + \frac{\Varx_{\bK \sim \mathcal{X}} \sbra{\ba_{\bK}}}{T}  \le \frac{4 \epsilon^\param}{C} + \frac{\epsilon^\param \cdot \Varx_{\bK \sim \mathcal{X}} \sbra{n-\ba_{\bK}}}{Cn}  \le  \frac{4 \epsilon^\param}{C} + \frac{\epsilon^\param \cdot \Ex_{\bK \sim \mathcal{X}} \sbra{(n-\ba_{\bK})^2}}{Cn} .
 \]
By \Cref{eq:total-inf-estimator} and \Cref{prop:poincare}, we have that $0 \le a_K \le n$  for any symmetric convex $K$.  Thus, we can further upper bound the right hand side to obtain 
\[
 \Var[\bM] \le \frac{4\epsilon^\param}{C} +\frac{ \epsilon^\param \cdot \Ex_{\bK \sim \mathcal{X}} [n-\ba_{\bK}]}{C}. 
 \]
Recalling from \Cref{equation:n-ak} that $\Ex_{\bK \sim \mathcal{X}} [n-\ba_{\bK}] \ge c\epsilon$, a routine computation shows that for a sufficiently large constant $C$, we have
 \[
 \Var[\bM] \le  \frac{4\epsilon^\param}{C} +\frac{ \epsilon^\param \cdot \Ex_{\bK \sim \mathcal{X}} \sbra{n-\ba_{\bK}}}{C} \le \frac{ \Ex_{\bK \sim \mathcal{X}} \sbra{n-c\epsilon/2-\ba_{\bK}}^2}{100}. 
 \]
\Cref{eq:truncated-mix-statistics1} and Chebyshev's inequality now give that when ${\cal D} = {\cal D}_{\mathcal{X}}$, 
with probability at least $9/10$ we have
$
\bM \le n-c\epsilon/2, 
$ completing the proof.
\end{proofof}


\section{An $O(n/\eps^2)$-Sample Algorithm for General Convex Sets}
\label{subsec:general-convex}

In this section we present a $O(n/\eps^2)$-sample algorithm for distinguishing the standard normal distribution from the standard normal distribution restricted to an arbitrary convex set. More precisely, we prove the following:

\begin{theorem} \label{thm:convex}
There is an algorithm, \Convex\ (\Cref{alg:convex}), with the following performance guarantee:  Given any $\eps > 0$ and access to independent samples from any unknown distribution ${\cal D} \in \Pconv$, the algorithm uses ${O}(n/\eps^2)$ samples, runs in $O(n^{2}/\eps^2)$ time, and

\begin{enumerate}

\item If ${\cal D}=N(0,I_n)$, then with probability at least $9/10$ the algorithm outputs ``un-truncated;''

\item If $\dtv({\cal D},N(0,I_n)) \geq \eps$, then with probability at least $9/10$ the algorithm outputs ``truncated.''

\end{enumerate}

\end{theorem}

\begin{algorithm}
\caption{Distinguisher for General Convex Sets}
\label{alg:convex}
\vspace{0.5em}
\textbf{Input:} $\calD\in\Pconv$, $\eps > 0$\\[0.5em]
\textbf{Output:} ``un-truncated'' or  ``truncated''

\ 

\Convex$(\calD, \epsilon)$:
\begin{enumerate}
	
    \item {If $\GMT(\calD, 0.01)$ returns ``reject,'' then halt and output\newline ``truncated.''}

    \item For $T =C \cdot n/\epsilon^2$, sample points $\x{1}, \ldots, \x{T} \sim \calD$. 
    
    \item Set $\bM := \frac{1}{T}\sum_{j=1}^T \|\x{j}\|^2$. 
    Output ``truncated'' if $\bM \leq n - c\eps/2$ and ``un-truncated'' \newline otherwise. 
    
\end{enumerate}

\end{algorithm}

Note that the estimator $\bM$ in \Cref{alg:convex} is identical to the estimator $\bM$ in \Cref{alg:symconvex} to distinguish Gaussians restricted to (mixtures of) symmetric convex sets. As we will see, the analysis of \Cref{alg:symconvex} via the \Poincare~inequality for convex influences (cf. \Cref{prop:poincare}) extends to arbitrary convex sets with ``large inradius.'' For the ``small inradius'' case, we further consider sub-cases depending on how close the center of mass of $\calD$, denoted $\mu$, is to the origin (see \Cref{fig:convex-ub}):
\begin{itemize}
	\item \textbf{Case 1:} When $\|\mu\| \gg 0$, we detect truncation via \GMT~in Step 2 of~\Cref{alg:convex}.
	\item \textbf{Case 2:} When $\|\mu\| \approx 0$, we show that we can detect truncation via $\bM$. This is our most technically-involved case and relies crucially on (small extensions of) Vempala's quantitative Brascamp-Lieb inequality (\Cref{prop:vempala}).
\end{itemize}

\subsection{Useful Preliminaries}
\label{subsec:prelims-for-general-convex}

Below are two useful consequences of Vempala's quantitative one-dimensional Brascamp-Lieb inequality (\Cref{prop:vempala}) which will be useful in our analysis of \Cref{alg:convex}. 

The following proposition says that if the center of mass of a convex body (with respect to the standard normal distribution) along a direction $v\in \S^{n-1}$ is the origin, then the convex influence of $v$ on the body  is non-negative.

\begin{proposition}
\label{cor:inf-nonneg-com}
	Given a convex set $K\sse\R^n$ and $v\in \S^{n-1}$, if 
	\[\Ex_{\bx\sim N(0,I_n)}\sbra{K(\bx)\abra{v,\bx}} = 0,\]
	then $\Inf_v[K] \geq 0$.
\end{proposition}

\begin{proof}
	We may assume without loss of generality that $v = e_1$. Note that the function $f :\R\to\R_{\geq 0}$ defined by
	\[f(x) := \Ex_{\by\sim N(0,I_{n-1})}\sbra{K(x, \by)},\]
	is a log-concave function (this is immediate from the Pr\'ekopa-Leindler inequality~\cite{Prekopa1973,Leindler1972}).
	Furthermore, note that by \Cref{fact:avg-norm-influence},
	\[\sqrt{2}\cdot\Inf_{v}[K] = \Ex_{\bx\sim N(0,1)}\sbra{f(\bx)(1-\bx^2)},\]
	and so the result follows by \Cref{prop:vempala}.
\end{proof}

We also require a version of \Cref{prop:vempala} for log-concave functions whose center of mass with respect to the standard normal distribution is not at the origin. Looking ahead, \Cref{prop:vempala-shifted} will come in handy when analyzing \Cref{alg:convex} for Gaussians restricted to convex sets with small inradius and with center of mass close to the origin. 

\begin{proposition}
\label{prop:vempala-shifted}	
	Let $f:\R\to\R_{\geq 0}$ be a one-dimensional log-concave function with 
	\[\Ex_{\bx\sim N(0,1)}\sbra{\bx f(\bx)} = \Ex_{\bx\sim N(0,1)}\sbra{\mu\cdot f(\bx)}\]
	for some $\mu \in \R$.
	Then 
	\[\Ex_{\bx\sim N(0,1)}\sbra{\bx^2 f(\bx)} \leq \pbra{1 + \mu^2}\cdot\Ex_{\bx\sim N(0,1)}\sbra{f(\bx)}.\]
	Furthermore, if $\supp(f)\sse(-\infty,\epsilon]$, then 
	\begin{equation} \label{eq:potato}
		\Ex_{\bx\sim N(0,1)}\sbra{\bx^2 f(\bx)} \leq \pbra{1 + \mu^2 - \frac{1}{2\pi}e^{-(\eps-\mu)^2}}\cdot\Ex_{\bx\sim N(0,1)}\sbra{f(\bx)}.
	\end{equation}
	
\end{proposition}

We prove \Cref{prop:vempala-shifted} by translating the log-concave function $f$ so that its center of mass (with respect to a shifted Gaussian) is the origin, and then appealing to \Cref{prop:vempala}.

\begin{proof}
	Note that it suffices to prove \Cref{eq:potato}. Consider the one-dimensional log-concave function $\wtf:\R\to\R_{\geq  0}$ given by 
	\[\wtf(x) := f(x+\mu).\]
	It is clear that $\supp(\wtf) \sse (-\infty, \eps-\mu]$ if $\supp(f) \sse (-\infty, \eps]$. Note that 
	\begin{equation} \label{eq:potato-means}
		\Ex_{\bx\sim N(-\mu,1)}\sbra{\wtf(\bx)} = \int_{\R} f(x + \mu)\varphi(x+\mu)\,dx = \Ex_{\bx\sim N(0,1)}\sbra{f(\bx)}.
	\end{equation}
	We also have that 
	\begin{align*}
		\Ex_{\bx\sim N(-\mu,1)}\sbra{\bx\wtf(\bx)} & = \int_{\R} x f(x+\mu)\phi(x+\mu)\, dx\\
		&= \int_{\R} (y-\mu) f(y)\phi(y)\, dy\\
		&= \Ex_{\by\sim N(0,1)}\sbra{\by f(\by)} - {\Ex_{\by\sim N(0,1)}\sbra{\mu\cdot f(\by)}}\\
		&= 0,
	\end{align*}
	where we made the substitution $y = x-\mu$. Therefore, by \Cref{prop:vempala}, we have that 
	\begin{equation}
		\label{eq:potato-1}
		\Ex_{\bx\sim N(-\mu,1)}\sbra{\bx^2\wtf(\bx)} \leq \pbra{1 - \frac{1}{2\pi}e^{-(\eps-\mu)^2}}\cdot\Ex_{\bx\sim N(-\mu,1)}\sbra{\wtf(\bx)}. 
	\end{equation}
	However, we have 
	\begin{align}
		\Ex_{\bx\sim N(-\mu,1)}\sbra{\bx^2\wtf(\bx)} &= \int_\R x^2 f(x+\mu)\phi(x+\mu)\, dx \nonumber \\
		&= \int_\R (y-\mu)^2 f(y)\phi(y)\,dy \nonumber \\
		&= \Ex_{\by\sim N(0,1)}\sbra{\by^2f(\by)} - \Ex_{\by\sim N(0,1)}\sbra{\mu^2\cdot f(\by)}. \label{eq:potato-2}
	\end{align}
	\Cref{eq:potato} now follows from \Cref{eq:potato-means,eq:potato-1,eq:potato-2}.
\end{proof}

\subsection{Proof of \Cref{thm:convex}}
\label{subsec:general-convex-pf}

We can now turn to the proof of \Cref{thm:convex}.

\begin{proof}[Proof of \Cref{thm:convex}]
	Suppose first that $\calD = N(0,I_n)$. In this case, 
	\begin{equation} 
	\label{eq:gaussian-M-mean}
		\E\sbra{\bM} = \frac{1}{T}\sum_{j=1}^T \E\sbra{\|\x{j}\|^2} = \frac{1}{T}\sum_{j=1}^T n = n.
	\end{equation}
	We also have that 
	\begin{align}
		\Var\sbra{\bM} &= \frac{1}{T^2}\sum_{j=1}^T\Var\sbra{\|\x{j}\|^2} 
		= \frac{1}{T}\pbra{\Varx_{\bx\sim N(0,I_n)}\sbra{\|\bx\|^2}} 
		= \frac{1}{T}\sum_{i=1}^n\Varx_{\bx_i\sim N(0,1)}\sbra{\bx_i^2} 
		= \frac{2n}{T}, 	\label{eq:gaussian-M-variance}	
	\end{align}
	where we used the fact that $\Varx_{\bx\sim N(0,1)}[\bx^2] = 2.$ Looking ahead, we also note that in this case, by \Cref{prop:DKP} we have the algorithm outputs ``truncated'' in Step 2 with probability at most $0.01$. 
	
	Next, suppose that $\calD = N(0,I_n)|_K$ for convex $K\sse\R^n$ with $\dtv(\calD,N(0,I_n))\geq\eps$. Let us write $\rinn$ for the in-radius of $K$. Suppose first that $\rinn \geq 0.1$. In this case, we have that 
	\begin{equation} 
	\label{eq:big-inrad-M-mean}
		\E\sbra{\bM} = \Ex_{\bx\sim\calD}\sbra{\|\bx\|^2} \leq n -\Omega\pbra{\eps}.
	\end{equation}
	by \Cref{eq:tv-vol}, \Cref{fact:avg-norm-influence}, and \Cref{prop:kkl}. By independence of the $\x{j}$'s, we also have that 
	\[\Var[\bM] = \frac{1}{T^2}\sum_{j=1}^T \Varx_{\x{j}\sim\calD}\sbra{\|\x{j}\|^2}.\]
	Note, however, that by \Cref{prop:BL} we have
	\begin{equation}
	\label{eq:big-inrad-M-variance}
		\Varx_{\bx\sim\calD}\sbra{\|\bx\|^2} \leq 4\Ex_{\bx\sim\calD}\sbra{\|\bx\|^2} \qquad\text{and so}\qquad \Var[\bM] \leq \frac{4n}{T},
	\end{equation}
	where the second inequality follows from \Cref{eq:big-inrad-M-mean}. From \Cref{eq:gaussian-M-mean,eq:big-inrad-M-mean}, we have that the means of $\bM$ under $N(0,I_n)$ versus $N(0,I_n)|_K$ differ by $\Omega(\eps)$, and from \Cref{eq:gaussian-M-variance,eq:big-inrad-M-variance} we have that the standard deviations in both settings are on the order of $O(\sqrt{n/T})$. 
	This shows that \Convex~indeed succeeds in distinguishing $\calD = N(0,I_n)$ from $\calD = N(0,I_n)|_K$ with $O(n/\eps^2)$ samples in the case that $\rinn \geq 0.1.$
		
	For the rest of the proof we can therefore assume that $\rinn < 0.1$. It follows from the hyperplane separation theorem that there exists $x^\ast \in \S^{n-1}(0.1)$ such that $K$ lies entirely on one side of the hyperplane that is tangent to $\S^{n-1}(0.1)$ at $x^\ast$.  Recalling that the standard normal distribution is invariant under rotation, we can suppose without generality that $x^\ast$ is the point $(0.1, 0^{n-1})$, so we have that either
	\[
	K\sse \{x \in \R^n : x_1 < 0.1\} \qquad\text{or}\qquad 
	K\sse \{x \in \R^n : x_1 \geq 0.1\},
	\]
	corresponding to (a) and (b) respectively in \Cref{fig:convex-ub}. Writing $\mu$ for the center of mass of $\calD$, i.e. 
	\[\mu := \Ex_{\bx\sim\calD}\sbra{\bx},\] 
	we can apply another rotation to obtain $\mu = (\mu_1,\mu_2,0^{n-2})$ while maintaining that $x^\ast = (0.1, 0^{n-1})$. 
	Now we consider two cases based on the norm of $\mu$:
	
	\medskip
	
	\textbf{Case 1.} If $\|\mu\|^2 \geq 0.06$, then we claim that Step 1 of \Cref{alg:convex} will correctly output ``truncated'' with probability at least 99/100. Indeed, it is readily verified that $\calD$ is log-concave and so this follows immediately from~\Cref{prop:DKP}.
	
%
%
	\bigskip
	
	\textbf{Case 2.} If $\|\mu\|^2 < 0.06$, then we will show that \Cref{alg:convex} will output ``truncated'' with probability at least $9/10$ in Step 3. We will do this by proceeding analogously to the ``large inradius'' ($\rinn \geq 0.1$) setting considered earlier. Recall that 
	\begin{equation}
		\label{eq:recall}
		\Ex\sbra{\bM} = \sum_{i=1}^n \Ex_{\bx\sim\calD}\sbra{\bx_i^2}.
	\end{equation}
	For $i \in \{3,\ldots, n\}$, as $\mu_i = 0$, we have by \Cref{cor:inf-nonneg-com} that $\Inf_i[K]\geq 0$, and so 
	\begin{equation} \label{eq:non-neg-inf-for-3-to-n}
		\Ex_{\bx\sim\calD}\sbra{\bx_i^2} \leq 1 \qquad\text{for } i \in\{3,\ldots,n\}
	\end{equation}
	by \Cref{fact:avg-norm-influence}. 
	

\begin{figure}[t]

\centering
%
%
%

\begin{tikzpicture}[scale=0.9]

\path[pattern=north west lines,pattern color=black!25] (2,3) -- (4,3) -- (4,-3) -- (2,-3);
\draw (2,-3)--(2,3);	
\filldraw[black] (0,0) circle (1pt) node[anchor=north east]{$(0,0)$};
\filldraw[black] (2,0) circle (1pt) node[anchor=north east]{$(0.1,0)$};

{
\draw (-0.1, -3) -- (-0.1, 3);
\draw (0.1, -3) -- (0.1, 3);
\path[fill=purple!20!blue!40!white,opacity=0.6] (-0.1,-3) -- (-0.1,3) -- (0.1,3) -- (0.1,-3);
}

\draw (-4,0)--(4,0) node[right]{$x$};
\draw (0,-3)--(0,3) node[above]{$y$};

\node (a-blah) at (0,-3.5) {(a)};
\node (a-K) at (-0.5,2) {$K$};

\end{tikzpicture}	\hfill\hfill \begin{tikzpicture}[scale=0.9]

{\path[pattern=north west lines,pattern color=black!25] (2,3) -- (-2,3) -- (-2,-3) -- (2,-3);}

{\draw (2,-3)--(2,3);}
\filldraw[black] (0,0) circle (1pt) node[anchor=north east]{$(0,0)$};
{\filldraw[black] (2,0) circle (1pt) node[anchor=north east]{$(0.1,0)$};}

{\filldraw[rotate around={45:(0,0)}, fill=purple!20!blue!40!white,opacity=0.6] (2.75,-2.25) ellipse (1 and 1);}

{\filldraw[black] (3.5, 0.35) circle (1pt) node[anchor=south]{$(\mu_1,\mu_2)$};}

\draw (-2,0)--(6,0) node[right]{$x$};
\draw (0,-3)--(0,3) node[above]{$y$};

\node (b-blah) at (2,-3.5) {(b)};
\node (b-K) at (4.5,1.5) {$K$};

\end{tikzpicture}

\caption{The ``small inradius'' ($\rinn \leq 0.1$) setting in the analysis of \Cref{alg:convex}, with $\mu$ denoting the center of mass of $K$. Our estimator for (a) is $\mathrm{Avg}\pbra{\|\x{j}\|^2}$, whereas for (b) we simply estimate $\mu$.}
\label{fig:convex-ub}
	
\end{figure}
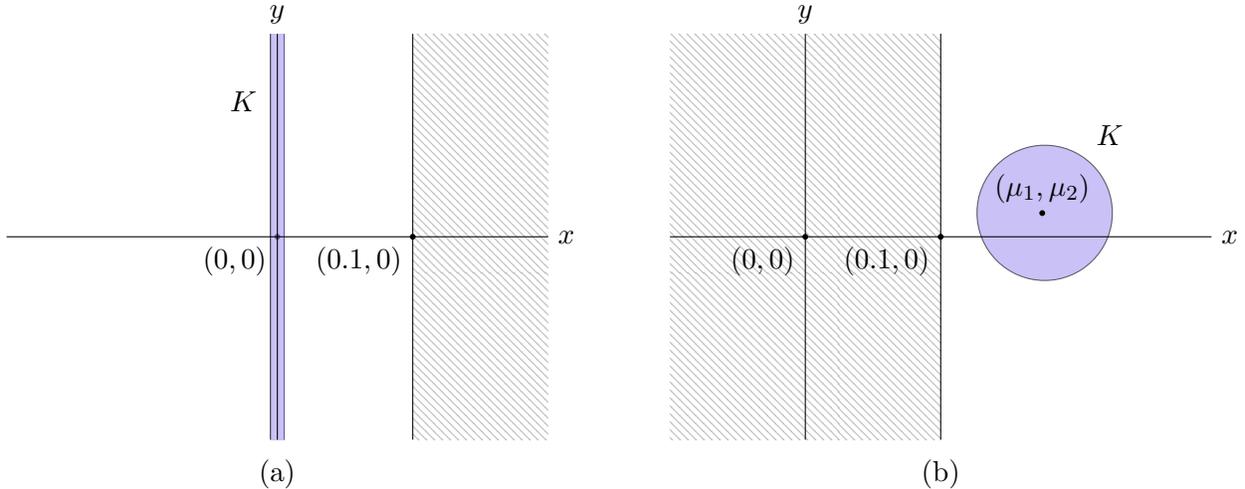
	
	We now consider coordinates 1 and 2. Consider the one-dimensional log-concave functions $f_1,f_2:\R\to\R_{\geq 0}$ defined by
	\[f_1(x) := \Ex_{\by\sim N(0,I_{n-1})}\sbra{K(x,\by)} \qquad\text{and}\qquad f_2(x) := \Ex_{\by\sim N(0,I_{n-1})}\sbra{K(\by_1, x,\by_2,\ldots, \by_{n-1})}.\]
	Note that $\E[f_1] = \E[f_2] = \vol(K)$. It is also immediate that 
	\begin{equation}
		\label{eq:vempala-shifted-var}
		\Ex_{\bx\sim\calD}\sbra{\bx_i^2} = \frac{\Ex_{\bx \sim N(0,1)}\sbra{\bx^2 f_i(\bx)}}{\vol(K)}.
	\end{equation}
	Since we have 
	\[\Ex_{\bx \sim N(0,1)}\sbra{\bx f_1(\bx)} = \mu_1\cdot\vol(K)\qquad\text{and}\qquad \Ex_{\bx \sim N(0,1)}\sbra{\bx f_2(\bx)} = \mu_2\cdot\vol(K),\]
it follows from \Cref{prop:vempala-shifted} that
	\begin{equation}
		\label{eq:inf12-lbs}
		\frac{\Ex_{\bx \sim N(0,1)}\sbra{\bx^2 f_1(\bx)}}{\vol(K)} \leq 1 + \mu_1^2 - \frac{1}{2\pi}e^{-(0.1-\mu_1)^2}
		\qquad \text{and}\qquad 
		\frac{\Ex_{\bx \sim N(0,1)}\sbra{\bx^2 f_2(\bx)}}{\vol(K)} \leq 1+\mu_2^2
	\end{equation}
(note that we used the fact that $\supp(f_1) \sse (-\infty, 0.1]$ in the first inequality above). Combining \Cref{eq:vempala-shifted-var,eq:inf12-lbs} and recalling that $\|\mu\|^2 < 0.06,$ we get that 
	\begin{equation}
		\label{eq:first-two-coords}
		\Ex_{\bx\sim\calD}\sbra{\bx_1^2 + \bx_2^2} \leq 2 + \|\mu\|^2 - \frac{1}{2\pi}e^{-(0.1-\mu_1)^2} 
		<2.06 - \frac{1}{2\pi}e^{-(0.1+\sqrt{0.06})^2}
		< 1.95.
	\end{equation}
	Combining \Cref{eq:recall,eq:non-neg-inf-for-3-to-n,eq:first-two-coords}, we get that 
	\begin{equation} \label{eq:apple}\Ex[\bM] = \Ex\sbra{\|\bx\|^2} \leq n - 0.05.
	\end{equation}
	As in \Cref{eq:big-inrad-M-variance}, by the Brascamp-Lieb inequality (\Cref{prop:BL}) we have that 
	\begin{equation} \label{eq:banana}\Varx[\bM] \leq \frac{4n}{T},\end{equation}
	and so by \Cref{eq:apple}, \Cref{eq:banana} and Chebyshev's inequality, for a suitable choice of $C$ algorithm
	\Convex~will output ``truncated'' in Step~3(a) with probability at least 0.9.
	\end{proof}

\newcommand{\z}{{\bz}}

\section{Lower Bound for Testing Convex Truncation}
\label{sec:chi-lb}

In this section, we present our lower bound for testing convex truncation. Our lower bound is information-theoretic and applies to all algorithms, computationally efficient or otherwise. Formally, we prove the following:

\begin{theorem} \label{thm:chi-lb}
Let $\eps \in (0, 1/2]$. Let $A$ be any algorithm which is given access to samples from an unknown distribution ${\cal D}$ and has the following performance guarantee:
\begin{enumerate}
\item If ${\cal D}=N(0,I_n)$, then with probability at least $9/10$ the algorithm outputs ``un-truncated'';
\item If ${\cal D} \in \Psymm$ and has $\dtv({\cal D},N(0,I_n))\geq \eps$ then with probability at least $9/10$ the algorithm outputs ``truncated.''
\end{enumerate}
Then, $A$ must draw $T = \Omega\pbra{n/\epsilon}$ samples from $\calD$.
\end{theorem}

As $\Psymm \sse \Pconv$ and $\Psymm \sse \Mix$, \Cref{thm:chi-lb} immediately that the algorithms~\Convex~and \symmconvex~are essentially optimal in terms of sample complexity for testing convex truncation and truncation by a mixture of symmetric convex sets respectively.

\subsection{Useful Preliminaries}
\label{subsec:chi-prelims}

We refer the reader to~\Cref{appendix:hermite} for background on Hermite analysis. 

%

\begin{claim}
\label{clm:GVunitvectors}
Suppose $|\delta| \le 1$, and let $(\bX, \bY)$ be $\delta$-correlated mean-0 Gaussians, i.e. 
\[
	(\bX, \bY) \sim N\pbra{
	\begin{bmatrix}
		0 \\ 0
	\end{bmatrix},
	\begin{bmatrix}
		1 & \delta \\ \delta & 1
	\end{bmatrix}
	}. 
\]
Let 
\begin{equation} \label{eq:kappa-def}
	\kappa := \Phi^{-1}\pbra{1-\frac{\eps}{2}}, \qquad\text{so}\qquad\Prx_{\bg \sim N(0,1)}\sbra{|\bg|\le\kappa}=1-\eps.
\end{equation}
For all $\eps \in (0,1)$, we have

\[
	\Pr\sbra{|\bX| \le \kappa ~\text{and}~ |\bY| \le \kappa} \leq
(1-\eps)^2 + \delta^2\cdot\eps.
\]

%
%
\end{claim}

\begin{proof}
	Let $K\sse\R$ be the symmetric interval $K := [-\kappa, \kappa]$; we also write $K: \R \to \{0,1\}$ to denote the indicator function of this interval. We have 	\begin{align}
		\Pr\sbra{|\bX| \le \kappa ~\text{and}~ |\bY| \le \kappa} 
		&= \Ex_{\bg\sim N(0,1)}\sbra{K(\bg)\cdot\U_\delta K(\bg)} \nonumber \\
		&= \sum_{\alpha\in\N} \delta^{|\alpha|} \bW^{|\alpha|}[K] \nonumber \\ 
		&= (1-\eps)^2 + \sum_{i=1}^\infty \delta^{i} \bW^{=i}[K] \nonumber \\ 
		&= (1-\eps)^2 + \delta^2\pbra{\sum_{i = 1}^\infty \delta^{2(i-1)} \bW^{=2i}[K]} \label{eq:ty-symmetry}
	\end{align}
where \Cref{eq:ty-symmetry} uses the fact that $K$ is an even (i.e.~symmetric) function and hence has no Hermite weight on odd levels.

To conclude, note that 
\[
\sum_{i = 1}^\infty \delta^{2(i-1)} \bW^{=2i}[K]
\leq \sum_{i=1}^\infty \bW^{=2i}[K]
=\sum_{i=1}^\infty \bW^{=i}[K]
=\Var[K]=\eps \cdot (1-\eps) < \eps,
\]
where the first equality again uses that $K$ is an even function.
%
\end{proof}

An equivalent rephrasing of \Cref{clm:GVunitvectors} is as follows:

\begin{corollary} \label{cor:GVunitvectors}
Let $v,w \in \mathbb{S}^{n-1}$ and let $\kappa,\eps$ be as above. 
Then for all $\eps \in (0,1)$, we have
\[
	\Prx_{\bg \sim N(0,1)^n}[|v \cdot \bg| \le \kappa \ \textrm{and} \ |w \cdot \bg| \le \kappa] \le (1-\eps)^2 + \eps \cdot \langle v,w \rangle^2.
\]
%
%
%
\end{corollary}

%
%
We will also require the following standard concentration bound:

\begin{fact}[Lemma~2.2 of~\cite{Ball1997}] \label{fact:subgaussian}
	Fix $v\in\S^{n-1}$. For $0\leq \eps < 1$, we have 
	\[
		\Prx_{\bw\sim\S^{n-1}}\sbra{\abra{\bw, v}\geq \eps} \leq \exp\pbra{\frac{-n\eps^2}{2}}.
	\]
\end{fact}

\subsection{Proof of~\Cref{thm:chi-lb}}
\label{subsec:proof-of-chi-lb}

We start by defining the family of truncations of $N(0, I_n)$ by symmetric convex sets that we will employ in our lower bound: 

\begin{definition} \label{def:dno}
	For $v\in \S^{n-1}$, define $K_v$ as the set
	\[
	 	K_v := \cbra{x \in \R^n : |x\cdot v|\leq \kappa}
	\]
	where $\kappa = \kappa(\eps)$ is as in~\Cref{eq:kappa-def}, and let $\calD_v := N(0, I_n)|_{K_v}$. 
\end{definition}

For notational convenience, let $P$ be the distribution on $\R^{nT}$ induced by first drawing a Haar-random vector $\bv\sim \S^{n-1}$ and then drawing $T$ i.i.d.~samples from $\calD_{\bv}$. 
Similarly, let $Q$ be the distribution on $\R^{nT}$ induced by drawing $T$ i.i.d.~samples from $N(0, I_n)$. 

It is clear from~\Cref{def:dno} that $\vol(K_v) = 1-\eps$ for all $v\in\S^{n-1}$. Consequently, we get
\[
	\pbra{\frac{dP}{dQ}(x^{(1)}, \ldots, x^{(T)})}^2 = \pbra{\frac{1}{1-\eps}}^{2T}\Ex_{\bv, \bw \sim \mathbb{S}^{n-1}} \sbra{ \prod_{i=1}^T K_{\bv}(x^{(i)}) \prod_{i=1}^T K_{\bw}(x^{(i)})},
\]
and so 
\begin{align}
	\dchi(P\,\|\,Q) 
	&= \Ex_{\x{i}\sim N(0, I_n)}\sbra{\pbra{\frac{1}{1-\eps}}^{2T}\Ex_{\bv, \bw \sim \mathbb{S}^{n-1}} \sbra{ \prod_{i=1}^T K_{\bv}(\x{i}) \prod_{i=1}^T K_{\bw}(\x{i})}} - 1 \nonumber \\
	&= \Ex_{\bv, \bw \sim \mathbb{S}^{n-1}}\sbra{\pbra{\frac{1}{1-\eps}}^{2T}\Ex_{\x{i}\sim N(0, I_n)} \sbra{ \prod_{i=1}^T K_{\bv}(\x{i}) \prod_{i=1}^T K_{\bw}(\x{i})}} - 1 \nonumber \\
	&= \Ex_{\bv, \bw \sim \S^{n-1}}\sbra{\pbra{\frac{1}{1-\eps}}^{2T}\Ex_{\bx\sim N(0, I_n)}\sbra{K_{\bv}(\bx)K_{\bw}(\bx)}^T} - 1\label{eq:chi-calc-independence} \\
	&= \Ex_{\bv, \bw \sim \S^{n-1}}\sbra{\Delta(\bv, \bw)^T} - 1, \nonumber
\end{align}
where $\Delta(v, w)$ is defined as
\begin{equation} \label{eq:Delta-def}
	\Delta(v, w) := \pbra{\frac{1}{1-\eps}}^{2}\Ex_{\bx\sim N(0, I_n)}\sbra{K_{v}(\bx)K_{w}(\bx)}. 
\end{equation}
Note that~\Cref{eq:chi-calc-independence} relied on the independence of $\x{i}\sim N(0, I_n)$. 

Let $T$ be as in~\Cref{thm:chi-lb}. It follows from~\Cref{cor:GVunitvectors} and the above that 
\[
	\Delta(v, w) \leq 1 + \frac{\eps\cdot\abra{v,w}^2}{(1-\eps)^2},
\]
and so 
\begin{align}
	\dchi(P\,\|\,Q) 
	&\leq \Ex_{\bv, \bw\sim \S^{n-1}}\sbra{\pbra{1 + \frac{\eps\cdot\abra{\bv,\bw}^2}{(1-\eps)^2}}^T - 1} \nonumber \\
	&= \int_{t = 0}^\infty \Prx_{\bv, \bw\sim \S^{n-1}}\sbra{\pbra{1 + \frac{\eps\cdot\abra{\bv,\bw}^2}{(1-\eps)^2}}^T - 1 \geq t}\,dt \nonumber \\ 
	&= 0.1 + \int_{t=0.1}^\infty \Prx_{\bv, \bw\sim \S^{n-1}}\sbra{\pbra{1 + \frac{\eps\cdot\abra{\bv,\bw}^2}{(1-\eps)^2}}^T - 1 \geq t}\,dt \nonumber \\
	&\leq 0.1 + \int_{t=0.1}^\infty \Prx_{\bv, \bw\sim \S^{n-1}}\sbra{\exp\pbra{\frac{T\cdot\eps\cdot\abra{\bv,\bw}^2}{(1-\eps)^2}} - 1 \geq t} \,dt \label{eq:ruggles} \\
	&= 0.1 + \int_{t=0.1}^\infty \Prx_{\bv, \bw\sim \S^{n-1}}\sbra{\abra{\bv,\bw}^2 \geq \frac{(1-\eps)^2\ln(1+t)}{T\cdot \eps}} \,dt \nonumber \\
	&\leq 0.1 + \int_{t=0.1}^\infty 2\exp\pbra{\frac{-n(1-\eps)^2\ln(1+t)}{2T \eps}}\,dt \label{eq:subgaussian} \\
	&= 0.1 + \int_{t=0.1}^\infty 2(1+t)^{-\frac{n(1-\eps)^2}{2T\eps}}\, dt \nonumber \\
	&\leq 0.11 \nonumber
\end{align}
where~\Cref{eq:ruggles} relied on the inequality $1+x\leq \exp(x)$, \Cref{eq:subgaussian} relied on~\Cref{fact:subgaussian}, and the final inequality relied on $T = \Omega(n/\eps)$. \Cref{thm:chi-lb} now follows from the above with $T = T_1$ thanks to~\Cref{eq:chi-to-tv} and~\Cref{fact:dtv-distinguishing}. 
\qed 

\section*{Acknowledgements}

A.D. is supported by NSF grants CCF-1910534, CCF-1926872, and CCF-2045128. S.N. is supported by NSF grants  IIS-1838154, CCF-2106429, CCF-2211238, CCF-1763970, and CCF-2107187. R.A.S. is supported by NSF grants  IIS-1838154, CCF-2106429, and CCF-2211238. This material is
based upon work supported by the National Science Foundation under grant numbers listed above.
Any opinions, findings and conclusions or recommendations expressed in this material are those of
the authors and do not necessarily reflect the views of the National Science Foundation (NSF). 

This work was partially completed while S.N.~was a participant in the program on ``Meta-Complexity'' at the Simons Institute for the Theory of Computing. 

\bibliography{allrefs.bib}
\bibliographystyle{alpha}

\appendix


\section{Hardness for Mixtures of General Convex Sets} \label{sec:no-generalization}

\Cref{thm:symmetric-convex-mixture-informal} gives an efficient ($O(n)$-sample) algorithm that distinguishes $N(0,I_n)$ from $N(0,I_n)$ conditioned on a mixture of (any number of) symmetric convex sets, and \Cref{thm:general-convex-informal} gives an efficient ($O(n)$-sample) algorithm that distinguishes $N(0,I_n)$ from $N(0,I_n)$ conditioned on any single convex set (which may not be symmetric). We observe here that no common generalization of these results, to mixtures of arbitrary convex sets, is possible with any finite sample complexity, no matter how large:

\begin{theorem} \label{thm:no-generalization}
Let $\Mixgeneral$ denote the class of all convex combinations (mixtures) of distributions from $\Pconv$, and let $N$ be an arbitrarily large integer ($N$ may depend on $n$, e.g.~we may have $N=2^{2^n}$).
For any $0 < \eps < 1$, no $N$-sample algorithm can successfully distinguish between the standard $N(0,I_n)$ distribution and an unknown distribution ${\cal D} \in \Mixgeneral$ which is such that $\dtv(N(0,I_n),{\cal D}) \geq \eps.$
\end{theorem}

\noindent
\emph{Proof sketch:} The argument is essentially that of the the well-known $\Omega(\sqrt{L})$-sample lower bound for testing whether an unknown distribution over the discrete set $\{1,\dots,L\}$ is uniform or $\Omega(1)$-far from uniform \cite{GRdist:00,BFRSW13}.
Let $M=\omega({\frac {N^2} {1-\eps}})$, and consider a(n extremely fine) gridding of $\R^n$ into disjoint hyper-rectangles $R$ each of which has $\vol(R)=1/M$.  (For convenience we may think of $M$ as being an $n$-th power of some integer, and of $\eps$ as being of the form $1/k$ for $k$ an integer that divides $M$.) We note that for any set $S$ that is a union of such hyper-rectangles, the distribution $N(0,I_n)|_S$ is an element of $\Mixgeneral$.  

Let $\bS$ be the union of a  random collection of exactly $(1-\eps)M$ many of the hyper-rectangles $R$. We have $\vol(\bS)=(1-\eps) M$, so $\dtv(N(0,I_n),N(0,I_n)|_{\bS})=\eps$, and consequently a successful $N$-sample distinguishing algorithm as described in the theorem must be able to distinguish $N(0,I_n)$ from the distribution ${\cal D}=N(0,I_n)|_{\bS}$.  But it is easy to see that any $o(\sqrt{(1-\eps)M})$-sample algorithm will, with $1-o(1)$ probability, receive a sample of points that all come from distinct hyper-rectangles; if this occurs, then the sample will be distributed precisely as a sample of the same size drawn from $N(0,I_n).$
\qed


\section{Hermite Analysis over $N(0, I_n)$}
\label{appendix:hermite}

Our notation and terminology follow Chapter~11 of~\cite{ODonnell2014}. 
We say that an $n$-dimensional \emph{multi-index} is a tuple $\alpha \in \N^n$, and we define 
\[
|\alpha| := \sum_{i=1}^n \alpha_i.
\]

For $n \in \N_{>0}$, we write $L^2(\R^n)$ to denote the space of functions $f: \R^n \to \R$ that have finite second moment under the Gaussian distribution, i.e. $f\in L^2(\R^n)$ if 
\[
\|f\|^2 = \Ex_{\bx \sim N(0,I_n)} \left[f(\bx)^2\right] < \infty.
\]
We view $L^2(\R^n)$ as an inner product space with 
\[\la f, g \ra := \Ex_{\bx \sim  N(0,I_n)}[f(\bx)g(\bx)]\]
We recall the Hermite basis for $L^2(\R, \gamma)$:

\begin{definition}[Hermite basis]
	The \emph{Hermite polynomials} $(h_j)_{j\in\N}$ are the univariate polynomials defined as
	$$h_j(x) = \frac{(-1)^j}{\sqrt{j!}} \exp\left(\frac{x^2}{2}\right) \cdot \frac{d^j}{d x^j} \exp\left(-\frac{x^2}{2}\right).$$
\end{definition}

The following fact is standard:

\begin{fact} [Proposition~11.33 of~\cite{ODonnell2014}] \label{fact:hermite-orthonormality}
	The Hermite polynomials $(h_j)_{j\in\N}$ form a complete, orthonormal basis for $L^2(\R)$. For $n > 1$ the collection of $n$-variate polynomials given by $(h_\alpha)_{\alpha\in\N^n}$ where
	$$h_\alpha(x) := \prod_{i=1}^n h_{\alpha_i}(x)$$
	forms a complete, orthonormal basis for $L^2(\R^n)$. 
\end{fact}

Given a function $f \in L^2(\R^n)$ and $\alpha \in \N^n$, we define its \emph{Hermite coefficient on} $\alpha$ as $\wh{f}(\alpha) = \la f, h_\alpha \ra$. It follows that $f:\R^n\to\R$ can be uniquely expressed as 
\[
    f = \sum_{\alpha\in\N^n} \wh{f}(\alpha)h_\alpha
\]
with the equality holding in $L^2(\R^n, \gamma)$; we will refer to this expansion as the \emph{Hermite expansion} of $f$. One can check that Parseval's and Plancharel's identities hold in this setting:
\[\abra{f,f} = \sum_{\alpha\in \N^n}\wh{f}(\alpha)^2 \qquad\text{and}\qquad \abra{f,g} = \sum_{\alpha\in \N^n}\wh{f}(\alpha)\wh{g}(\alpha).\]
It is also readily verified that $\E[f(\bx)] = \wh{f}(0^n)$ and $\Var[f(\bx)] = \sum_{\alpha\neq 0^n} \wh{f}(\alpha)^2$ where ${\bx\sim N(0,I_n)}$. 
We will write $\bW^{=k}[f]$ for the \emph{Hermite weight of $f$ at level-$k$}, i.e. 
\[
    \bW^{=k}[f] := \sum_{|\alpha| = k} \wh{f}(\alpha)^2,
\]
with $\bW^{\leq k}[f]$ defined similarly. 

\begin{definition} [Ornstein-Uhlenbeck semigroup] \label{def:ou-operator}
	Let $\rho \in [0,1]$. The \emph{Ornstein-Uhlenbeck operator} $\U_\rho$ is defined by its action on $f \in L^2(\R^n)$ as follows:  
	\[
		\U_\rho f(x) := \Ex_{\bg \sim N(0, I_n)}\left[f(\rho x + \sqrt{1-\rho}\bg)\right].
	\]
\end{definition}

The Ornstein-Uhlenbeck semigroup is sometimes referred to as the family of \emph{Gaussian noise operators} or \emph{Mehler transforms}. The Ornstein-Uhlenbeck semigroup acts on the Hermite expansion as follows:

\begin{fact} [Proposition~11.33, \cite{ODonnell2014}] \label{fact:gaussian-noise-expansion}
	For $f \in L^2(\R^n, \gamma)$, the function $\U_\rho f$ has Hermite expansion
	$$\U_\rho f = \sum_{\alpha\in\N^n}\rho^{|\alpha|}\wh{f}(\alpha)h_\alpha.$$
\end{fact}

\end{document}